\documentclass[sigconf]{acmart}
\AtBeginDocument{%
  }

\setcopyright{acmlicensed}
\copyrightyear{2026}
\acmYear{2026}
\acmDOI{XXXXXXX.XXXXXXX}
\acmISBN{978-1-4503-XXXX-X/2026/06}

\usepackage{braket}
\usepackage{xcolor}
\usepackage[table]{xcolor}
\usepackage{amsthm}
\usepackage{graphicx}
\usepackage{subcaption}
\usepackage{forest}
\usepackage{tikz}
\usetikzlibrary{positioning,arrows.meta,calc}
\newtheorem{proposition}{Proposition}

\ifdefined\SHOWCOMMENTS

    \newcommand{\anik}[1]{%
        \textcolor{blue}{\textbf{Anik:} #1}%
    }

    \newcommand{\suman}[1]{%
        \textcolor{red}{\textbf{Suman:} #1}%
    }

    \newcommand{\siyi}[1]{%
        \textcolor{cyan}{\textbf{Siyi:} #1}%
    }

\else

    \newcommand{\anik}[1]{}
    \newcommand{\suman}[1]{}
    \newcommand{\siyi}[1]{}

\fi

\begin{document}


\title{%
{\normalsize Invited Paper}\\
Structure-Aware Placement and Routing of Multi-Controlled Toffoli on Bivariate Bicycle Code Architectures}


\author{Anik Basu Bhaumik}
\orcid{0009-0000-0778-2234}
\affiliation{%
  \institution{Nanyang Technological University}
  \country{Singapore}
}
\email{anikbasu001@e.ntu.edu.sg}

\author{Suman Dutta}
\orcid{0009-0001-4891-9919}
\affiliation{%
\institution{Nanyang Technological University}
 \country{Singapore}
 }
 \email{sumand.iiserb@gmail.com}

\author{Siyi Wang}
\orcid{0009-0006-9128-1857}
\affiliation{%
  \institution{Nanyang Technological University}
  \country{Singapore}
  }
  \email{siyi002@e.ntu.edu.sg}

\author{Anupam Chattopadhyay}
\orcid{0000-0002-8818-6983}
\affiliation{%
  \institution{Nanyang Technological University}
  \country{Singapore}
  }
\email{anupam@ntu.edu.sg}

\renewcommand{\shortauthors}{Basu Bhaumik et al.}

\begin{abstract}

The multi-controlled Toffoli (MCT) gate is a fundamental primitive in quantum circuit design, with applications in quantum arithmetic, cryptanalysis, and algorithmic implementations. Being a high-level logical operation, the efficient decomposition of MCT gates into lower-level netlists has remained a major optimization challenge for decades. While emerging quantum error-correcting codes such as bivariate bicycle (BB) codes drastically reduce fault-tolerance overhead, realizing non-Clifford circuits on modular BB-code architectures introduces complex compilation bottlenecks governed by inter-module routing, factory density, and layout. Consequently, the mapping of MCT gates onto BB-code architectures remains relatively unexplored.

In this paper, we overcome these challenges by mapping optimal-Toffoli-depth MCT decompositions (Dutta et al., PRA, 2025) onto BB-code-based fault-tolerant architectures via direct $\ket{\mathrm{CCZ}}$ state injection from an external magic state factory. We introduce a targeted placement strategy that exploits the binary-tree structure of MCT decompositions to co-locate interacting subtrees. This approach reduces inter-module instruction counts by up to $\mathbf{16.02}\%$ compared to a naive sequential first-fit placement. We also evaluate the impact of factory placement across different topologies, demonstrating that grid-based layouts yield up to a $\mathbf{23.7}\%$ reduction in inter-module instructions relative to linear architectures (Yoder et al., arXiv, 2025). Finally, we validate the practical viability of our compiled circuits by analyzing aggregate execution errors and logical failure probabilities using the bicycle-ISA error estimator \texttt{bicycle\_numerics} provided by the Qiskit community\footnote{\url{https://github.com/qiskit-community/bicycle-architecture-compiler}}.

\end{abstract}



\received{xx xxx 2026}
\received[revised]{xx xxx 2026}
\received[accepted]{xx xxx 2026}

\maketitle
\section{Introduction}
\label{sec:intro}
Quantum gates are the fundamental building blocks of quantum circuits~\cite{gates1995}. Unlike classical gates, quantum gates are inherently reversible and are mathematically represented by unitary matrices. Among them, the multi-controlled Toffoli (MCT) gate is the most significant one with widespread applications in arithmetic circuit design~\cite{arithmetic}, reversible computation, oracle construction~\cite{dutta2026qic}, and quantum error-correction subroutines. Since MCT gates are logical quantum operations, they must be decomposed into a universal gate set~\cite{universal-gate2005}, such as the Clifford+T gate set. Consequently, efficient decomposition of MCT gates into the Clifford+T gate set has received considerable attention over the past two decades. Existing approaches have optimized various logical-level metrics, including T-count~\cite{gidney2021cccz}, T-depth~\cite{dutta2025pra}, ancilla count~\cite{nie2024arxiv}, and overall gate complexity, yielding several asymptotically optimal constructions under unrestricted~\cite{khattar2025quantum,dutta2025pra} and restricted~\cite{bhaumik2026optimal} qubit connectivity, assuming ideal, error-free execution. However, the efficient implementation of MCT gates in an error-corrected architecture and their resource constraints remain relatively underexplored in the literature.

Quantum error correction~\cite{shor-code,qecc-25years} is widely regarded as an enabling technology for large-scale quantum computation. Although noisy intermediate-scale quantum (NISQ)~\cite{preskill2018quantum} devices have demonstrated remarkable experimental progress, their limited coherence times and high physical error rates prevent the reliable execution of deep quantum circuits. Fault-tolerant quantum computing~\cite{ftqc} overcomes these limitations by encoding logical qubits into quantum error-correcting codes (QECCs), thereby enabling arbitrarily long computations provided that the physical error rate remains below the fault-tolerance threshold. Consequently, the practical cost of executing a quantum algorithm is determined not only by its logical circuit complexity, but also by the substantial overhead introduced by fault-tolerant error correction. Accordingly, recent resource-estimation studies increasingly account for these error-correction costs when evaluating the feasibility of large-scale quantum algorithms~\cite{gidney2025factor,cain2026shor}.

In an error-corrected architecture, every logical operation must be performed on encoded qubits, which internally involve syndrome extraction, logical gate execution, and decoder-assisted error recovery. As a result, the physical resources required to implement an error-corrected MCT circuit depend not only on its logical decomposition but also on the underlying quantum error-correcting code, the realization of logical Clifford+T gates, decoder complexity, and hardware communication constraints. Therefore, logical optimization metrics such as T-count and T-depth alone are insufficient to capture the practical cost of large-scale MCT implementations. In this work, we aim to bridge this gap by evaluating these implementations using BB-code-specific resource metrics.

The surface code~\cite{surface-code} has long been a leading candidate for fault-tolerant quantum computation owing to its high error threshold and compatibility with local interactions. Nevertheless, its substantial physical-qubit overhead has motivated the development of quantum low-density parity-check (qLDPC) codes~\cite{qLDPC}, which offer asymptotically higher encoding rates while preserving sparse stabilizer measurements. Among these, BB codes are the most promising for fault-tolerant architectures due to their structured connectivity and favorable code-implementation trade-offs. These developments motivate a re-evaluation of quantum resource estimation for MCT gates in the qLDPC-based architectures.

\subsection{Organization and contribution}
\label{sub:org}
In this work, we address the lack of fault-tolerant resource estimates for Toffoli mapping by evaluating multi-controlled Toffoli (MCT) gates on bivariate bicycle (BB) code architectures. We decompose MCT gates into Toffoli gates and develop mapping techniques that exploit the decomposition's binary-tree structure. Our main contributions are as follows:
\begin{itemize}
    \item We establish a baseline comparison with a naive first-fit sequential placement of qubits relative to our structure-aware technique. We also do an architectural comparison between grids and linear architectures. We evaluate routing overhead using the inter-module instruction count, following \cite{sethi2026optimizing}. 
    
    \item We implement Toffoli gates via $\ket{\mathrm{CCZ}}$ injection, map the protocol to bicycle instructions including factory-to-module routing, estimate the resulting error with \texttt{bicycle\_numerics}, and model circuit failure using a Poisson distribution.
    \end{itemize}

%
\section{Background}
\label{sub:back}
In this section, we proceed with the preliminaries.
\subsection{Multi-controlled Toffoli gates}
The doubly controlled X gate, commonly known as the Toffoli gate, is one of the most important reversible operators in quantum computing. Acting on three qubits, the basic Toffoli gate flips the target qubit if and only if both control qubits are in the state $\ket{1}$. Its generalization, the multi-controlled Toffoli ($n$-MCT) gate, consists of $n$ control qubits and one target qubit and performs a conditional bit flip on the target when all control qubits are simultaneously in the state $\ket{1}$. For more details on quantum gates, we refer the reader to~\cite{nc}.

Over the past few decades, numerous efforts have been made to reduce the resource requirements for implementing MCT gates. For basic Toffoli gates, state-of-the-art techniques employ measurement-based uncomputation to achieve decompositions using only four T gates~\cite{gidney2018quantum, jaques2020eurocrypt}. On the other hand, recent works on MCT decomposition employ the conditionally clean ancilla technique~\cite{nie2024arxiv, khattar2025quantum, dutta2025pra} to reduce the ancilla count and the circuit depth. In this technique, the availability of clean ancillae is assumed to be limited; therefore, some of the working qubits are made to be conditionally clean and treated as ancillae during the decomposition. The technique was first introduced by Nie et al.~\cite{nie2024arxiv}, formalized by Khattar et al.~\cite{khattar2025quantum}, and later generalized in Dutta et al.~\cite{dutta2025pra}. Moreover, Dutta et al.~\cite{dutta2025pra} also provided an $\lceil\log_2 n\rceil$ Toffoli-depth decomposition of an $n$-MCT using $n-2$ additional ancilla qubits, and proved its optimality by showing that the Toffoli depth cannot be reduced further, irrespective of the number of available ancilla.

The optimal-Toffoli-depth decompositions of $n$-MCT gates, assuming all-to-all qubit connectivity, for $n=5$ is shown in Figure~\ref{fig:optimal_mct}. For a detailed summary of state-of-the-art results on decomposing MCT gates into basic Toffoli gates, we refer the reader to \cite[TABLE II]{dutta2025pra}.
\begin{figure}[htbp]
    \centering
    \scalebox{0.9}{\includegraphics[scale=0.8]{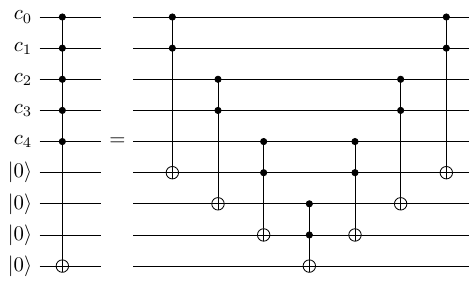}}
    \caption{Optimal Toffoli-depth $n$-MCT decompositions assuming all-to-all qubit connectivity, for $n=5$.}
    \label{fig:optimal_mct}
\end{figure}
\subsection{qLDPC and BB codes}
Quantum low-density parity-check (qLDPC) codes constitute an important family of stabilizer codes characterized by sparse parity-check matrices. Similar to their classical counterparts, qLDPC codes have the property that each stabilizer generator acts on only a constant number of qubits, while each qubit participates in a bounded number of stabilizer checks. This sparsity enables efficient syndrome extraction and, more importantly, allows qLDPC codes to achieve favorable encoding rates and code distances with substantially lower physical-qubit overhead than conventional topological codes. Recent constructions, including hypergraph-product codes~\cite{tillich2013tit}, quantum Tanner codes~\cite{leverrier2022focs}, balanced-product codes~\cite{bpqc2021}, and lifted-product codes~\cite{lpqc}, have demonstrated the possibility of constructing quantum codes with constant encoding rate and asymptotically growing distance.
\anik{I am thinking of removing the examples due to space constraints}

The sparse connectivity specified by their parity-check matrices can instead induce long-range interactions between physical qubits. Consequently, implementing qLDPC codes requires not only efficient syndrome extraction and decoding but also hardware architectures \cite{wang2025nature} and compilation strategies \cite{tourdegross}, similar to lattice surgery in surface codes ~\cite{gameofsurfacecodes,leblond2024realistic} that support their non-local connectivity. Among the various qLDPC constructions, bivariate bicycle (BB) codes are particularly relevant to practical fault-tolerant quantum computing. Recent experimental demonstrations have shown that BB codes can be implemented on superconducting quantum hardware by exploiting hardware capabilities that support long-range interactions~\cite{bravyi2024highthreshold}.

BB codes are CSS qLDPC codes constructed from two sparse bivariate polynomials $a(x,y)$ and $b(x,y)$ over $\mathbb{F}_2$. Evaluating these polynomials on cyclic-shift matrices $S_\ell, S_m$ provides two binary matrices $A = a(S_\ell,S_m),\, B = b(S_\ell, S_m)$, from which the CSS parity check matrices are defined as
$$H_X = [A \; B], \,H_Z = [B^T \; A^T].$$
The sparsity of $a$ and $b$ yields sparser stabilizer checks, thereby preserving the LDPC property. Unlike surface-code stabilizers, these interactions need not be geometrically local in a 2D nearest-neighbor layout. The cyclic-shift structure naturally induces structured long-range connectivity. BB codes also encode multiple logical qubits within a single code block, providing substantially higher encoding rates than conventional surface-code patches. In this work, we consider the $\left[[144,12,12]\right]$ BB-code instance, which encodes 12 logical qubits into 144 physical data qubits with distance 12. One logical qubit is reserved as a pivot/ancilla for logical operations and routing.

\subsection{Fault Tolerant compilation for BB codes}
Compilation for fault-tolerant architectures differs fundamentally from compilation for NISQ-based architectures~\cite{bhaumik2026optimal}. Rather than targeting individual physical qubits, fault-tolerant compilation operates at the logical level, where physical qubits are organized into encoded logical blocks. In the surface code, a logical qubit is typically represented by a code patch, whereas in BB-code-based architectures, a single code block or module can encode multiple logical qubits. Consequently, the target instruction set may also differ substantially from a conventional gate-based instruction set. A compiler may instead map the algorithm onto the fault-tolerant logical operations supported by the underlying architecture.

Another important aspect of fault-tolerant compilation is the placement and routing of magic-state factories~\cite{gameofsurfacecodes}. Non-Clifford operations, such as T gates, typically require the preparation and injection of magic states to implement the corresponding logical operation. Magic-state factories continuously produce these resource states, introducing an additional architectural component that must be considered during compilation. In particular, the compiler must account for the placement of factories and the communication required to deliver magic states to the logical qubits that consume them. In this paper, we consider a $\ket{\mathrm{CCZ}}$-state injection protocol for executing our Toffoli-intensive workloads. This approach allows the non-Clifford resource generation to be largely separated from the computation itself and enables optimizations at the logical level of the Toffoli network.

Yoder et al.~\cite{tourdegross} recently introduced a compilation scheme for BB-code-based quantum hardware. They developed a new instruction set architecture (ISA) for compiling quantum instructions to a fault-tolerant abstraction, together with corresponding logical error-rate estimates. These compilers currently support results through simulation as cloud-based FTQC hardware are not available yet. 


\section{Tree based Placement in BB Code Architecture}
\label{sec:1}
In this work, we abstract each decoded BB-code block as a logical module with capacity $C=11$ and study how Toffoli/CCZ interactions should be placed and routed across a graph of such modules. The BB-code-based architectures are not yet well established. Yoder et al.~\cite[Figure~2(a)]{tourdegross} considered a linear architecture in which BB modules are arranged sequentially, with a magic-state factory placed at one end. Sethi et al.~\cite{sethi2026optimizing} extended this architecture to a long-grid topology.

From~\cite[Table~2]{tourdegross}, it is evident that the cost of inter-module operations is substantially higher than that of intra-module operations. Hence, we use the number of inter-module interactions as our primary routing metric. The optimal Toffoli-depth decomposition of an $n$-MCT~\cite{dutta2025pra} has a hierarchical structure composed of smaller $k$-MCTs. We exploit this structure through a tree-based placement strategy that groups the $k$-MCTs corresponding to the same lower-level subtree into a single BB module. The intermediate ancilla qubits created by each group are then propagated toward the higher-level outputs. This preserves the hierarchical dependencies of the decomposition while reducing unnecessary inter-module communication. The grouping factor $k$ determines both the number of controls assigned to each module and the number of BB modules required, given by $\lceil  n/k \rceil$. Consequently, the choice of $k$ introduces a topology-dependent trade-off among local packing efficiency, communication distance, and residual ancilla capacity.

Each module has a maximum capacity of 11 logical qubits and can therefore accommodate at most $k=6$ control qubits (as the optimal Toffoli-depth $n$-MCT decomposition requires $2n-1$ working qubits~\cite{dutta2025pra}). However, accounting for the additional ancilla required to compose the complete MCT may make such a choice impractical. In particular, it may require additional uncomputation, increasing the gate complexity, or leave insufficient capacity for dirty ancillae. We therefore derive a bound on the permissible value of $k$ in Proposition~\ref{prop:prop1}.

\begin{proposition}
\label{prop:prop1}
   Consider an $n$-MCT decomposition mapped onto BB-code-based modules with logical capacity $C=11$. Then, each module can implement a maximum of $k$-MCT for $k \leq 5$, where $\lceil  n/k \rceil$ modules provide sufficient residual logical-qubit capacity to accommodate the higher-level MCT operations in the decomposition without requiring additional ancilla qubits.
\end{proposition}
\begin{proof}
From~\cite{dutta2025pra}, the binary-tree-based decomposition of an $n$-MCT requires $2n-1$ working qubits in total, including $n-2$ ancilla. Since each BB-code-based module has logical capacity $C=11$, a single module can theoretically accommodate a $6$-MCT, which requires 11 qubits. However, using the full module capacity for the lower-level MCTs leaves insufficient capacity for the ancillae required to construct larger MCTs. To avoid additional inter-module operations and ancilla management, we restrict the maximum MCT size per module to $n=5$. Hence, $k\leq 5$.

Consequently, implementing an $n$-MCT with $n\gg 5$ requires at least $\lceil n/5 \rceil$ many BB-code-based modules to accommodate the decomposition without requiring additional ancilla qubits.
\end{proof}

We find that both the number and placement of magic-state factories significantly influence routing efficiency, as each gate requiring a magic-state resource must communicate with a factory. Consequently, the factory configuration affects routing at a finer-grained, gate-level scale.




\subsection{Linear Topology placement} 
\label{seb:lin-top}
We first consider a linear architecture in which the BB modules are arranged sequentially along a single line. This is one of the most commonly explored architectures. Magic-state factories are attached to the modules at regular intervals, and we vary the factory placement period to study the effects of factory density and placement on routing performance. Figure~\ref{fig:routing_volume_comparison} shows how different $k$-groupings affect the number of inter-module interactions under a first-fit-based placement strategy. Varying the factory-placement period also narrows the gap between the different groupings and the baseline.
\begin{figure}[htbp]
    \centering
\scalebox{0.7}{    \includegraphics[width=0.9\linewidth]{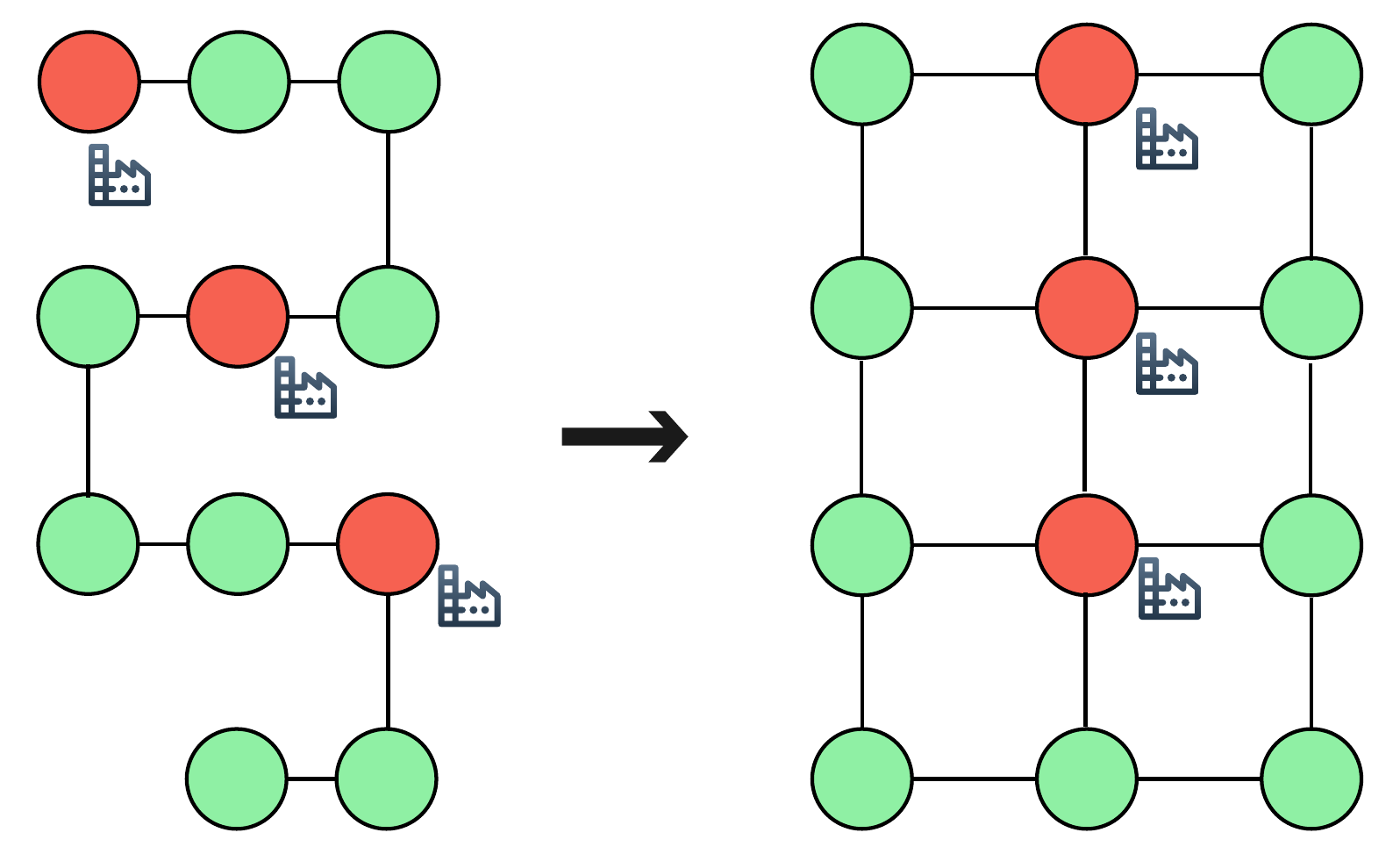}
}    \caption{Linear vs. grid factory placement, showing heuristic-selected factory modules in red and standard modules in green.}
    \label{fig:linvsgrid}
\end{figure}

\begin{figure*}[htbp]
    \centering

    
    \begin{subfigure}[t]{0.32\textwidth}
        \centering
        \includegraphics[width=\linewidth]
        {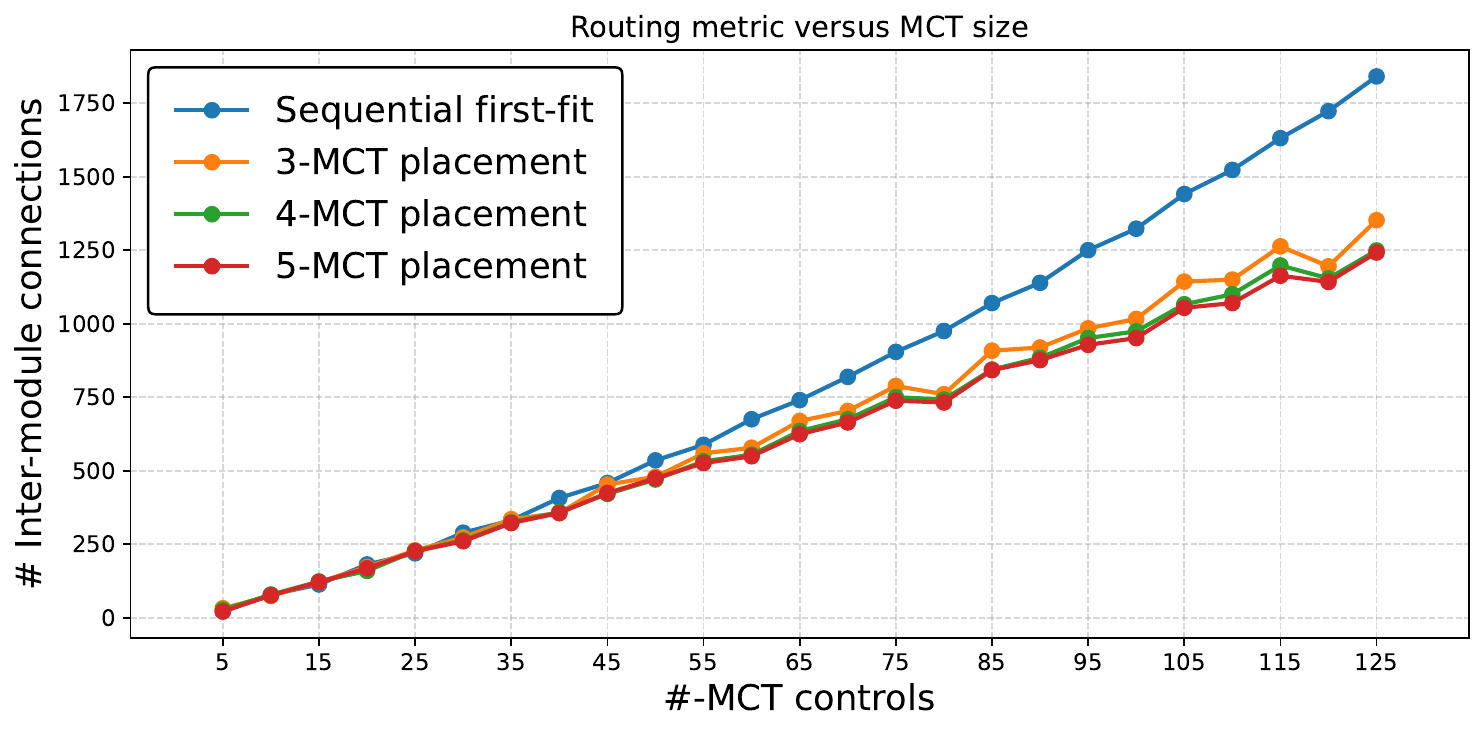}
        \caption{$\mathrm{Period}=2$}
        \label{fig:routing_period_2}
    \end{subfigure}
    \hfill
    \begin{subfigure}[t]{0.32\textwidth}
        \centering
        \includegraphics[width=\linewidth]
        {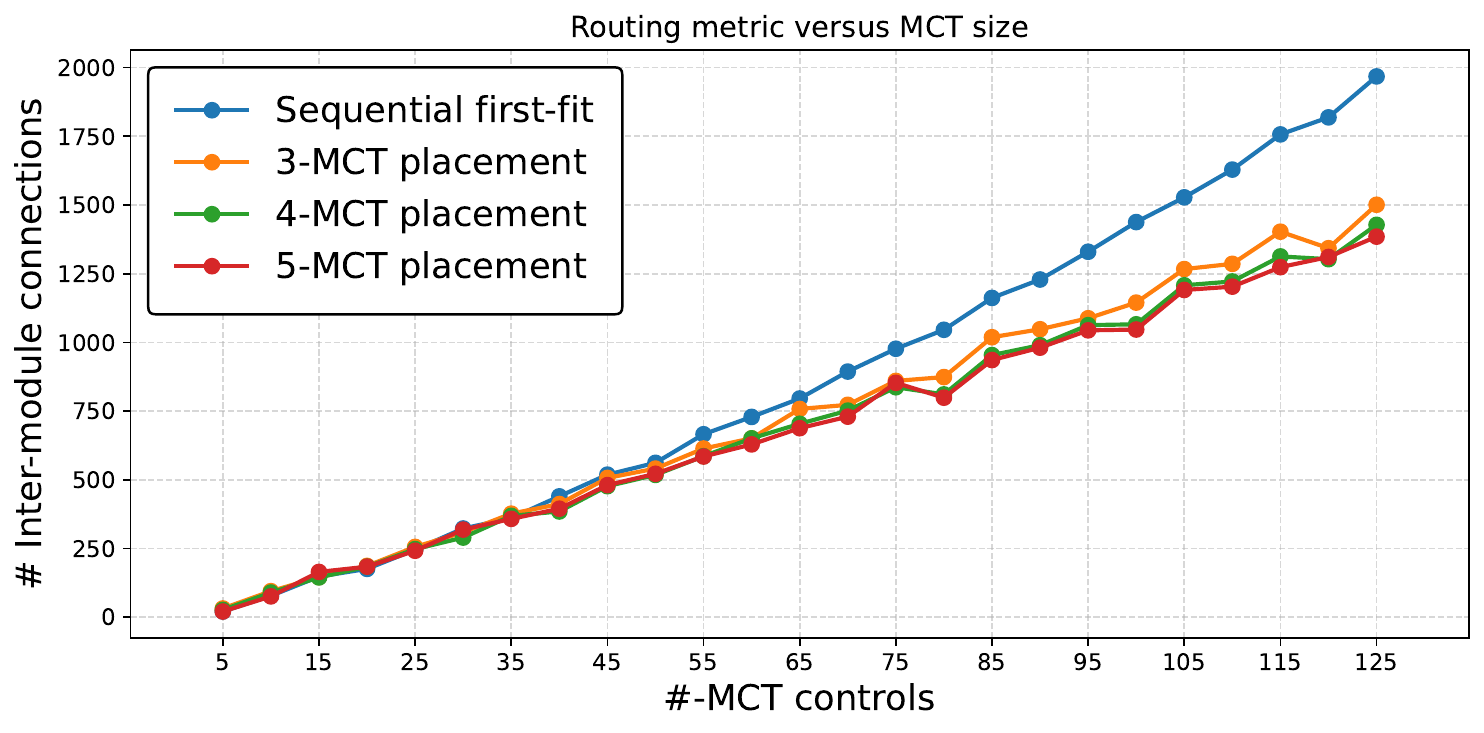}
        \caption{$\mathrm{Period}=3$}
        \label{fig:routing_period_3}
    \end{subfigure}
    \hfill
    \begin{subfigure}[t]{0.32\textwidth}
        \centering
        \includegraphics[width=\linewidth]
        {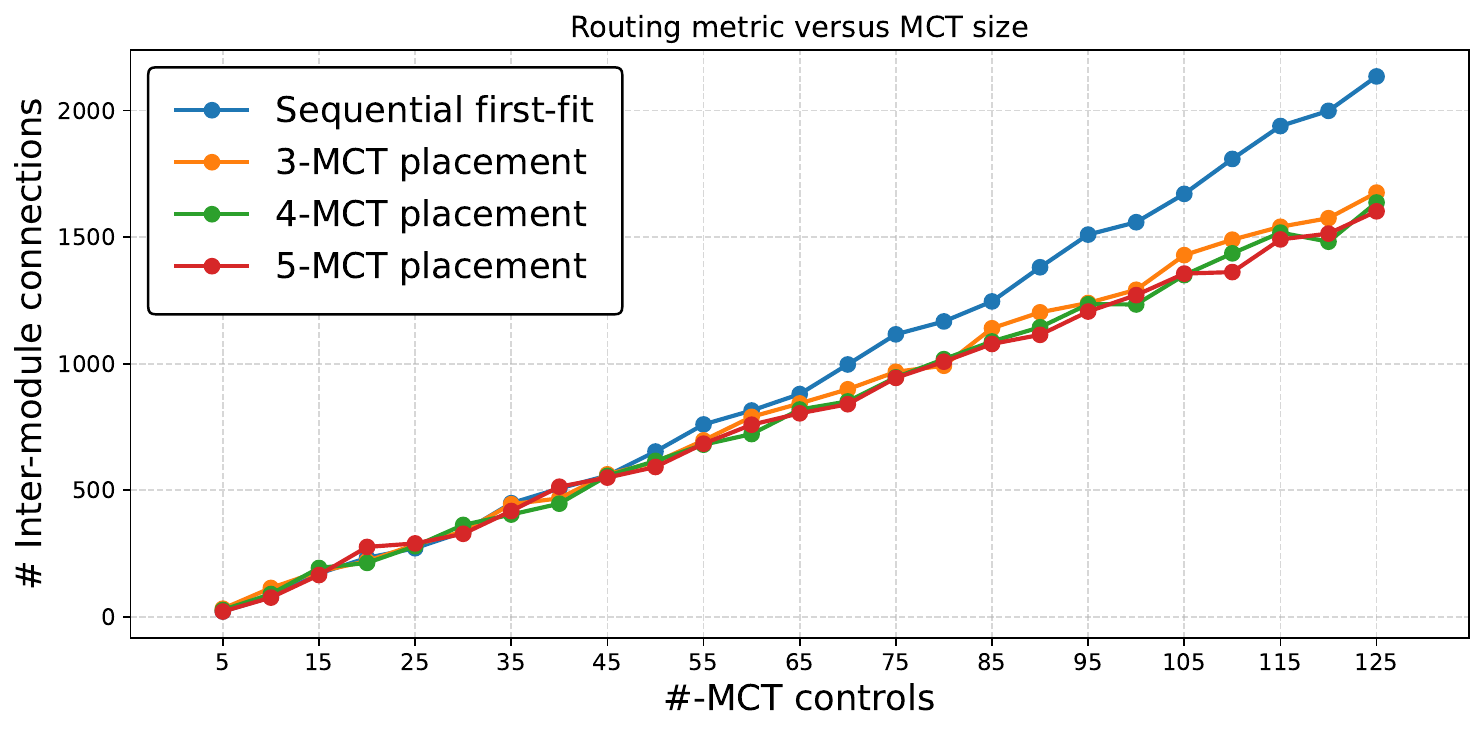}
        \caption{$\mathrm{Period}=4$}
        \label{fig:routing_period_4}
    \end{subfigure}

    \vspace{0.6em}


    \begin{subfigure}[t]{0.32\textwidth}
        \centering
        \includegraphics[width=\linewidth]
        {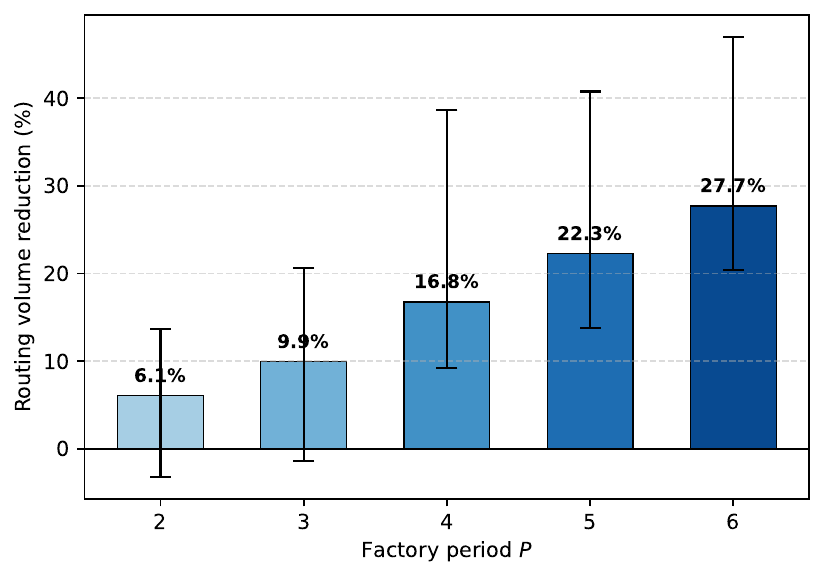}
        \caption{$k=3$}
        \label{fig:grid_routing_k3}
    \end{subfigure}
    \hfill
    \begin{subfigure}[t]{0.32\textwidth}
        \centering
        \includegraphics[width=\linewidth]
        {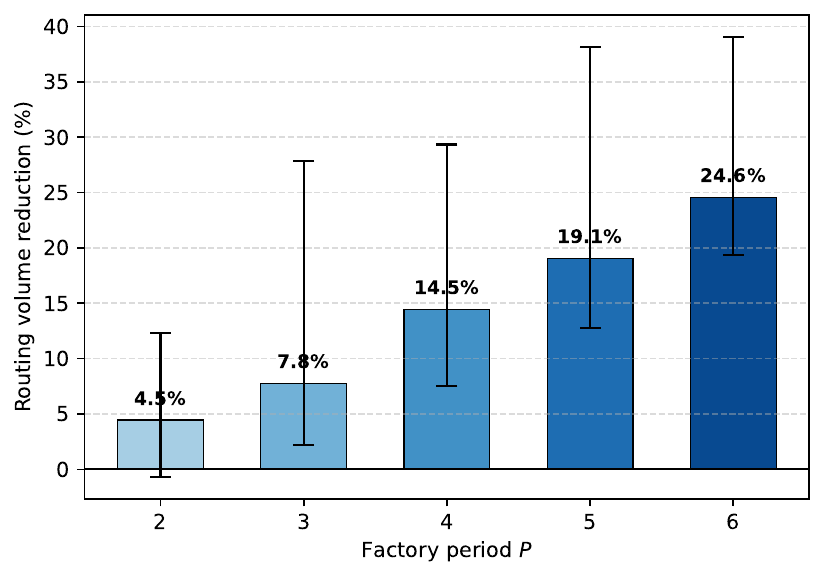}
        \caption{$k=4$}
        \label{fig:grid_routing_k4}
    \end{subfigure}
    \hfill
    \begin{subfigure}[t]{0.32\textwidth}
        \centering
        \includegraphics[width=\linewidth]
        {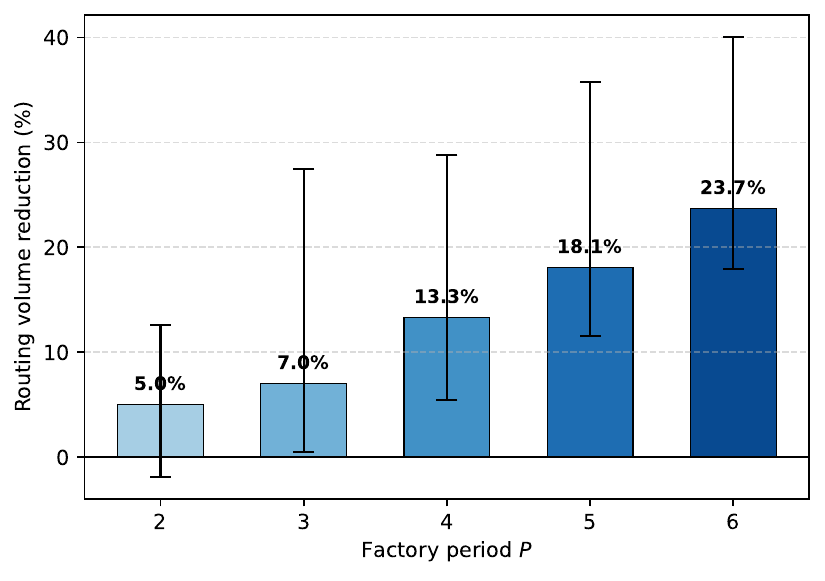}
        \caption{$k=5$}
        \label{fig:grid_routing_k5}
    \end{subfigure}

    \begin{subfigure}[t]{0.32\textwidth}
            \centering
            \includegraphics[width=\linewidth]
            {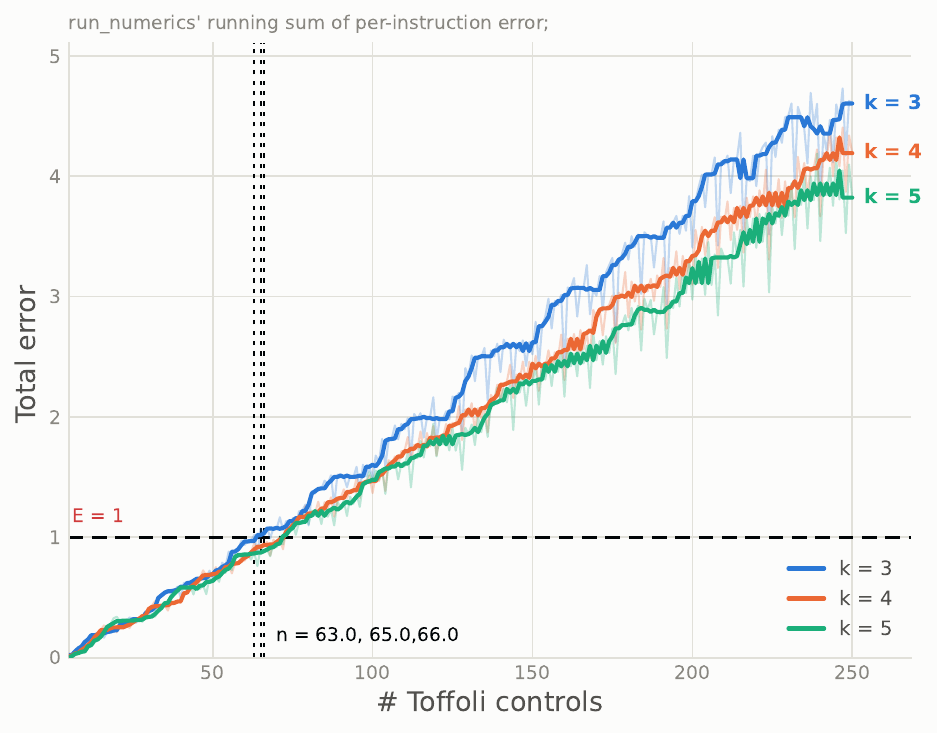}
            \caption{Aggregate logical error at $p=10^{-3}$.}
            \label{fig:gross_1e3_aggregate_error}
        \end{subfigure}
        \hfill
        \begin{subfigure}[t]{0.32\textwidth}
            \centering
            \includegraphics[width=\linewidth]
            {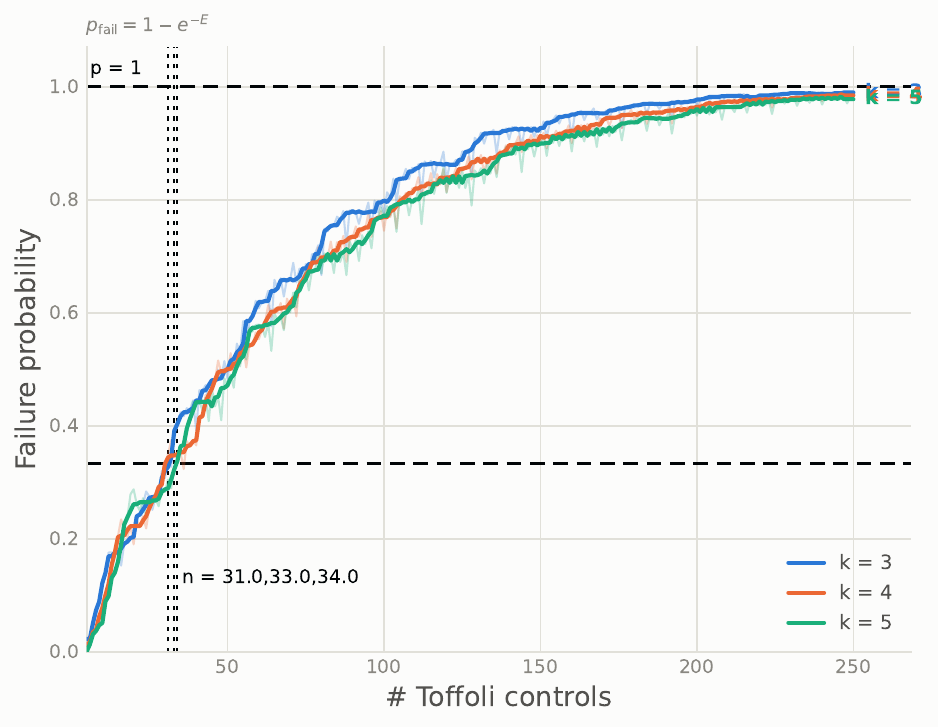}
            \caption{Circuit failure probability at $p=10^{-3}$.}
            \label{fig:gross_1e3_failure_probability}
        \end{subfigure}
        \hfill
        \begin{subfigure}[t]{0.32\textwidth}
            \centering
            \includegraphics[width=\linewidth]
            {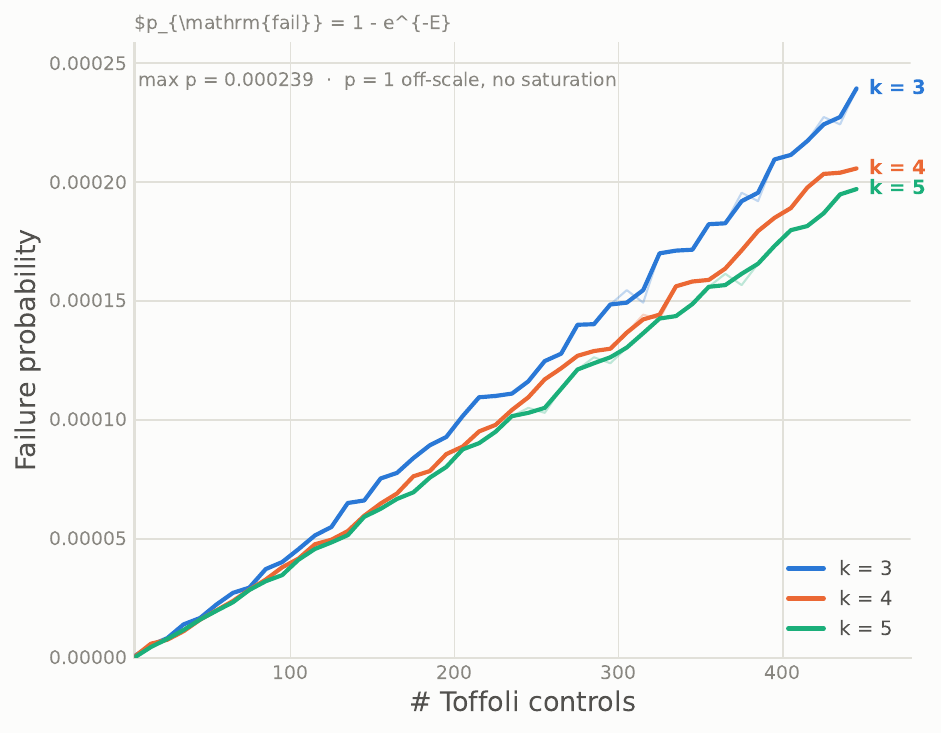}
            \caption{Circuit failure probability at $p=10^{-4}$.}
            \label{fig:gross_1e4_failure_probability}
        \end{subfigure}
    \caption{
        Inter-module routing characteristics under different factory
        placements and grouping configurations. 
        (a)--(c) show the linear BB-module architecture for factory
        periods $2$, $3$, and $4$, respectively, while
        (d)--(f) show the improvement in grid architecture vs linear for grouping factors $k=3$, $4$, and $5$.
        (g)--(h) shows the aggregate and circuit failure probabilities for hardware error rates of $p = 10^{-3}$ and and $p = 10 ^{-4}$
    }
    \label{fig:routing_volume_comparison}
\end{figure*}

\subsection{Grid Topology Placement}
\label{sub:grid-top}
We find that in the grid-based structure, factory density and factory placement take precedence on improving routing overhead. Sethi et al.~\cite{sethi2026optimizing} previously studied BB-code placement on grid topologies. However, their architecture places all the magic-state factories along a single edge of the grid, which does not adequately capture the distributed factory configurations considered in our work.

To investigate the impact of factory placement, we consider a square-grid architecture with dimensions $p = \sqrt{M}$ and $q = \lceil M/p \rceil$, where $M$ denotes the total number of BB modules. We compare the grid and linear topologies while keeping the number of magic-state factories fixed. For the grid topology, we use a facility-location heuristic to optimize factory placement and evaluate its impact on Toffoli routing performance, an example of such a placement is shown in Figure ~\ref{fig:linvsgrid}

\subsubsection{Facility Location Optimization Heuristic}
\label{subsub:floh}
This factory placement problem resembles facility location. We initialize factory sites using graph medians and refine them with swap-based local search. Each factory is placed in a separate module and connected to its assigned BB module by an additional edge, reflecting their treatment as distinct entities in fault-tolerant compilation. To determine the factory locations, we first employ a greedy heuristic~\cite{hakimi1964optimum}, followed by a 1-swap local-search heuristic~\cite{teitz1968heuristic}.

Let $G=(V, E)$ denote the BB-module graph, and let $d(u,v)$ denote the shortest-path distance between vertices $u,v\in V$. Given a set $S\subseteq V$ of factory sites, the greedy placement strategy minimizes the objective
$$J(S) = \sum_{v \in V} \min_{s\in S}d(v,s); \quad S \subseteq V,$$
where $S$ denotes the subset of modules selected as factory sites. After the greedy procedure has considered all vertices according to this objective, we perform a 1-swap local search over the candidate set. Specifically, each selected factory site $s \in S$ is tentatively replaced by an unselected module $v \in V \setminus S$. The replacement is accepted whenever $J\!\left((S \setminus \{s\}) \cup \{v\}\right) < J(S).$
The procedure terminates when no improving single-site swap exists.


\subsection{Evaluation of Routing Results}
\label{sub:evaluation}
We model the architectures as graphs and simulate qubit placement and instruction flow to quantify inter-module routing count and distance. The routing metric quantifies the number of inter-module interactions present, as previously modeled by \cite{sethi2026optimizing}. Specifically, for each interaction, we compute the distance between the factory and the modules containing the corresponding data qubits and sum these distances over all interactions. For example, consider a Toffoli gate $(c_1,c_2,t)$ and a $\ket{\mathrm{CCZ}}$ factory $F$. Since consuming a $\ket{\mathrm{CCZ}}$ state requires three independent Pauli-product measurements for injection. The number of inter-module instructions is given by
$$R = d(M(c_1),F) + d(M(c_2),F) + d(M(t),F)$$
where $M(q)$ denotes the module containing the qubit $q$. 

We compare our placement strategy against a first-fit baseline, in which the qubits are placed sequentially as shown in Figure~\ref{fig:optimal_mct}, following the order $c_1,c_2,\ldots , c_n,a_1,a_2,\ldots ,a_{n-2},t$. Then we compare grouping $k$-MCTs. Specifically, we partition $n$ control qubits into $\lceil n/k \rceil$ modules, form a $k$-MCT within each module, and ensure that the remaining ancillae can be accommodated within the residual capacity of each module in a first-fit fashion. We evaluate both the placement strategies for $n$ up to 127.

\subsubsection{Grouping and Factory-Period in Linear Topology}
Figures~\ref{fig:routing_period_2}, \ref{fig:routing_period_3}, and \ref{fig:routing_period_4} show the routing performance for different factory periods, where the factory period determines the spacing between consecutive factories along the linear topology. Across all three cases, the tree-based grouping strategy consistently improves the routing metric relative to the baseline. As the factory period increases, however, the gap between the baseline and the grouping strategy decreases, indicating that the efficiency of the linear topology deteriorates with increasing factory spacing. The maximum mean improvement is by $16.02\%$ for, $period = 2$. 

%
\subsubsection{Comparison of Grid Topology vs Linear Topology}
We compare the linear BB-module topology with a grid topology that has the same number of modules, $M$, and the same number of magic-state factories, as determined by the factory period, $f$. Figures~\ref{fig:grid_routing_k3}, \ref{fig:grid_routing_k4}, and \ref{fig:grid_routing_k5} show that relative improvement decreases modestly with k, from a maximum of $27.7\%$ at $k=3$ to $23.7\%$ at $k=5$. Tree-based placement has transferred inefficiently, since it assumes subtrees occupy the contiguous maximum first fit, which terms as locality on a chain but not on a grid. But, since $k$ is bounded in Proposition~\ref {prop:prop1}, these values cover the full set of admissible grouping factors ($k$), and the grid topology provides a positive impact throughout this range of $k = [3,5]$.

However, each of the figures also show that the improvement in the routing metric provided by the grid topology increases monotonically with $f$, ranging from $6.1\%$ to $27.7\%$ for $k=3$. In other words, the grid's advantage grows as factories thin out, this is where the linear's diameter starts to dominate. 
 
%
\section{Error Estimation for Routing}
\label{sec:2}
Error is the primary attribute that FTQC architectures claim to optimize, Yoder et al.~\cite[Table 2]{tourdegross} provides a bicycle ISA for compiling quantum circuits onto BB-code-based architectures. They also characterize the error rates associated with each of these operations. Hence, we develop a workflow that takes the architecture details of the mapped Toffoli circuit and its gate-DAG as input, generates a sequence of Bicycle-ISA instructions, and feeds the resulting instructions into the error estimator \texttt{bicycle\_numerics}. The Toffoli benchmarks use $\ket{\mathrm{CCZ}}$ injection  which is not natively supported by \texttt{bicycle\_compiler}~\cite{bicycle_compiler}.
. We therefore model the instructions required for the injection's communication and correction rounds explicitly for each Toffoli, and accumulate the resulting errors.
\paragraph{Factory distillation cost}
\label{subsub:fact_dist_cost}
The factory distillation cost corresponds to the cost of generating a $\ket{\mathrm{CCZ}}$ state in a magic-state factory. We assume this cost to be constant for a given factory, with the corresponding cost incurred each time a $\ket{\mathrm{CCZ}}$ state is produced. In this study, however, we omit the distillation cost and focus primarily on errors arising from routing and inter-module communication.

\paragraph{Injection cost}
\label{subsub:inj_cost}
The injection cost arises from entangling the $\ket{\mathrm{CCZ}}$ state with the data qubits and subsequently measuring the corresponding factory qubits. We model this process using \texttt{Measure} operations within the modules and \texttt{JointMeasure} operations across the logical qubits of different modules as shown in \cite{bicycle_compiler}

At the gate level, the injection procedure contains a $CNOT_{d_i\rightarrow r_i}$ followed by a measurement of the $r_i$ qubit, which is equivalent to a destructive two-qubit Pauli-product measurement, $Z_{d_i}Z_{r_i}$. For each Toffoli gate, we first select the nearest available factory by evaluating the distances between the factory and the modules containing the participating data qubits. After selecting the factory, we decompose the corresponding injection operations. For each $Z_{d_i}Z_{r_i}$ measurement, we consider two cases:
\begin{itemize}
	\item If $d_i$ is located in the same module to which the factory is attached, we model the operation as a local \texttt{Measure} operation within that module.
	\item If $d_i$ is located in a different module from the factory, we perform a local \texttt{Measure} operation in each module along the communication path between the factory module and the module containing $d_i$, and connect these operations using a sequence of \texttt{JointMeasure} instructions.
\end{itemize}
%
\paragraph{Correction cost}
\label{subsub:corr_cost}
Once the entanglement and measurement steps are complete, conditional Clifford corrections are required to obtain the desired output state up to a global phase. The correction cost depends on the measurement outcomes $m_0, m_1, m_2$ of the factory qubits. As discussed in~\cite[Section~3.4]{tourdegross}, Clifford corrections that cross module boundaries can only be implemented through cross-module measurements. The conditional corrections in the CCZ-injection protocol consist of $CZ$ and $Z$ operations. Since a $CZ$ correction may require an inter-module measurement, we estimate the expected correction cost by considering the probability that each conditional $CZ$ correction is required.

Let $C_0$, $C_1$, and $C_2$ denote the costs associated with the three possible $CZ$ corrections. Assuming that each measurement outcome $m_i$ is independently equal to 1 with probability $1/2$, the expected correction cost is .
$$\mathbb{E}\!\left[C_{corr}\right] = \mathbb{E}\!\left[m_0 C_0 + m_1 C_1 + m_2 C_2\right] = \frac{1}{2}C_0 + \frac{1}{2}C_1 + \frac{1}{2}C_2.$$

For each inter-module $CZ$ correction, we generate the corresponding \texttt{JointMeasure} operation with a weight of $1/2$ in the expected error calculation. Single-qubit $Z$ corrections are omitted because their cost within a module is negligible compared with that of inter-module operations and therefore contributes negligibly to the overall error estimate.

\subsection{Error Evaluation}
\label{sub:error_evaluation}

We report two error metrics. The first is the aggregate error, \[E= \sum_i p_i\]
Where $p_i$ denotes the logical error probability associated with the $i$-th instruction. The quantity $E$ represents the expected number of logical faults per Toffoli-circuit execution. Because it is an additive aggregate of individual instruction error probabilities, $E$ grows approximately linearly with the number of error-prone operations. Importantly, $E$ is an expected fault count rather than a probability and may therefore exceed unity. Next, we estimate the circuit failure probability as 
\[P_{fail} = 1 - e^{-E}\]
This expression follows from modeling the number of faults as a Poisson random variable with mean $E$, under the assumption that instruction failures are approximately independent and that the individual error probabilities are sufficiently small. 
Consequently, the reported figures quantify placement-dependent error rather than the end-to-end failure probability. We evaluate on the estimator provided by \cite{tourdegross,bicycle_compiler} which have encoded error rates for given instructions in the bicycle-ISA. The code for instruction generation is to be provided in GitHub \footnote{\url{https://github.com/anik314159/QLDPC-Toffoli}}.

\subsubsection{Results}
\label{subsub:result}

The error evaluations given are only for linear BB-code-based architectures. Although the Bicycle architecture permits more general module-based connectivity, extending the error analysis to the proposed two-dimensional grid would require additional assumptions regarding inter-module routing and physical connectivity. No significant improvement is seen over period, However in the aggregate error estimation (Figure \ref{fig:gross_1e3_aggregate_error}) the error-relations display decrease of error over increase of $k$ hence we chose an average value of $\textit{period}=4$, we report the aggregate error in Figure~\ref{fig:gross_1e3_aggregate_error}, and circuit failure probability for $p = 10^{-3}$ and $p = 10 ^{-4}$ Figure~\ref{fig:gross_1e3_failure_probability}, Figure~\ref{fig:gross_1e4_failure_probability}. Following ~\cite[Section~4]{tourdegross}, we use a circuit failure threshold of $1/3$. However, this threshold was originally reported for a complete computation, whereas we apply it here for two reasons: first, to enable direct comparison with~\cite{tourdegross}, and second, because inter-module measurements account for approximately $97\%$ of the total error in our model.

For a physical error rate of $p=10^{-3}$, the failure-probability threshold is reached at approximately $n=31$, $33$, and $34$ for $k=3$, $4$, and $5$, respectively.  For a physical error rate of $p=10^{-4}$, the projected circuit size required to reach the failure-probability threshold lies far beyond the range considered in our experiments $n \leq 450$, indicating that the threshold occurs only at substantially larger values of $n$.




\section{Conclusion and Future Works}
\label{sec:con}
Multi-controlled Toffoli gates provide an important benchmark for quantum circuit compilation and resource optimization. In this work, we study their mapping onto BB-Code based fault-tolerant architectures and introduce a structure-aware placement strategy that exploits the binary-tree structure of optimal-depth MCT decompositions. Our results demonstrate that decomposition-aware placement can reduce inter-module communication and improve routing efficiency. More broadly, our work highlights the shift in compilation objectives from the NISQ regime, where circuits are mapped directly onto physical qubits, to the fault-tolerant regime, where logical modules, communication constraints, and magic-state factories are central for error-safe circuit execution. Extending this approach to other arithmetic and cryptographic applications, as well as to other MCT decompositions, fault-tolerant topologies, and resource models, provides a natural direction for future research.



\bibliographystyle{ACM-Reference-Format}
\bibliography{Reference}

\appendix



\end{document}